\documentclass[12pt]{article}
\usepackage[utf8]{inputenc}

\usepackage{amsthm,amsmath,latexsym,amssymb,amsfonts,amscd}
\usepackage{graphics,lscape,fancyhdr,array,stmaryrd,euscript}
\usepackage{verbatim}
\usepackage{array}
\usepackage[all]{xy}

\usepackage[numbers,sort&compress]{natbib}
\usepackage[nottoc]{tocbibind}

\usepackage{natbib}
\usepackage{graphicx}
\usepackage{amsthm,amsmath,latexsym,amssymb,amsfonts,amscd,hyperref}

\usepackage[all]{xy}
\newtheorem{theorem}{Theorem}[section]

\newtheorem{proposition}[theorem]{Proposition}

\theoremstyle{definition}
\newtheorem{definition}[theorem]{Definition}
\newtheorem{example}[theorem]{Example}

\usepackage{xcolor}

\theoremstyle{remark}
\newtheorem{remark}[theorem]{Remark}

\numberwithin{equation}{section}

\newcommand{\lb}{[\![}

\newcommand{\rb}{]\!]}

\newcommand{\lbb}{\{\hspace{-3.1pt}[}

\newcommand{\rbb}{]\hspace{-3.2pt}\}}

\newcommand{\ml}{{\boldsymbol{\Lambda}}}

\setcitestyle{square}

\begin{document}

\title{\bf Lie and Leibniz Algebras \\of Lower-Degree Conservation Laws}
\author{Boris M. Elfimov \& Alexey A. Sharapov}

\date{{\it \small Department of Quantum Field Theory, Tomsk State University, \\ Lenin ave. 36, Tomsk 634050, Russia}}

\maketitle

\begin{abstract}
A relationship between the asymptotic and lower-degree conservation laws in (non-) linear gauge theories is considered. We show that the true algebraic structure underlying asymptotic charges is that of Leibniz rather than Lie. The Leibniz product is defined through the derived bracket construction for the natural Poisson brackets and the BRST differential. Only in particular,  though not rare, cases that the Poisson brackets of lower-degree conservation laws vanish modulo central charges,  the corresponding Leibniz algebra degenerates into a Lie one.  The general construction is illustrated by two standard examples: Yang-Mills theory and Einstein's gravity. 

\end{abstract}

\section{Introduction}
Although exact symmetries seldom if ever occur in nature, they do play a significant role in our understanding of physical laws. The famous Noether's theorem, for example, establishes a correspondence between global symmetries and conservation laws in Lagrangian theories. Local or gauge symmetries, in turn, govern the structure of fundamental interactions 
in the Standard Model of particle physics and General Relativity.  Unlike global symmetries, gauge 
invariance never leads to nontrivial conserved currents, as was found by Emmy Noether herself. Yet, 
in some cases, it is possible to associate certain conserved charges with gauge  symmetries as well. These are given by integrals over lower-dimensional surfaces rather than entire physical space. 
Therefore one calls them the lower-degree conservation laws. For a detailed account of the subject we refer the reader to  \cite{TSUJISHITA19913,10.2307/2152923,PhysRevLett.77.4109, BARNICH2000439,Sharapov:2016sgx} and references therein. 

It is not rare in field theory to consider a situation where the interaction of fields is presumably localized in a compact region of space, outside of which the fields behave as almost free. Such a supposition is a departing point of the scattering problem and an indispensable element of the multi-particle interpretation of quantum field states. In that case, one can compute the lower-degree conserved charges by integrating over surfaces lying entirely outside the interaction region. As the fields under the integral sign are supposed to be almost free, one may disregard much of the nonlinearities in the Lagrangian and charges.  Not only does this choice of integration surface facilitate the actual computation of charges, but it also  leads to a further generalization of the very notion of a conservation law. Indeed, if a nonlinear gauge theory admits a reasonable linearization, it may happen that the linear theory enjoys lower-degree conservation laws that are not present in  the full nonlinear theory. Then one may attribute the corresponding  lower-degree charges to the nonlinear theory itself or, more precisely, to those its solutions that differ arbitrary little from solutions of the linear theory outside a compact region of space (the interaction region).  In the physical literature, one usually refers to the charges assigned in such a way to a nonlinear gauge theory as {\it asymptotic charges}. The ADM energy in general relativity and color charges in Yang--Mills theory are prime examples of asymptotic charges defined through the integration over two-sphere at space infinity. For an extended discussion of the above  approach to asymptotic conservation laws see \cite{PhysRevLett.77.4109, Torre:1997cd}.

In Lagrangian theories, both the conventional and lower-degree conservation laws come equipped  with a
Lie algebra structure w.r.t. some natural Poisson brackets (canonical or covariant) \cite{Dickeybook,Barnich_1996,Sharapov_2015}.  Modulo central extension, this 
algebra is known to be isomorphic to the Lie algebra of infinitesimal symmetries.   If asymptotic charges are now defined through the linearization procedure above and do not survive in the full nonlinear theory, one runs into a problem: Linearization usually implies abelization of the gauge symmetry algebra, so that the corresponding asymptotic charges may form either an abelian Lie algebra or a central extension thereof, whatever a nonlinear gauge theory. This conclusion, however,  contradicts to some other approaches to asymptotic  conservation laws that do not appeal directly to linearization \cite{osti_4234816,Brown:1986nw,Silva_1999,BARNICH20023}. In the general case  one expects the nonabelian gauge symmetries to generate nonabelian Lie algebras of surface charges modulo central extension. We thus face the following dilemma: a nonabelian gauge theory {\it per se} may 
have no nontrivial conservation laws of lower degree, while the lower-degree conservation laws resulting from its linearization cannot form a nonabelian Lie algebra w.r.t. Poisson brackets.   

In this paper, we resolve this puzzle by introducing a new product on the space of lower-degree conservation laws of a linearized gauge theory. The product is constructed as a derived bracket \cite{Kosmann_Schwarzbach_2004} and involves the original Poisson  brackets together  with the classical BRST differential. It is the dependence of interaction through the BRST differential that restores a nonabelian algebra structure on conserved charges. Generally the product we introduce is not skew-symmetric but  satisfies the axioms of a Leibniz algebra.  This suggests that a `genuine' algebraic structure underlying asymptotic symmetries and  conservation laws is that  of Leibniz rather than Lie\footnote{An instructive discussion of the origin of Leibniz algebras in the context of gauge symmetries can be found in  \cite[Sec. 1]{Bonezzi_2020}.}. In order to test and illustrate our construction we re-derive the well-known algebras of surface charges in Einstein's gravity and Yang-Mills theory. 

Compared to other approaches our method is purely algebraic: we are not concerned with the fall-off
of  fields at infinity or suitable boundary conditions, nor do we make any assumption about the differential structure of field equations and/or gauge generators. What we actually use is the separation of field equations into a free part and interaction. The separation is clearly ambiguous and we  regard it as part of the definition of a classical field theory. Another advantage of our construction is that it applies equally well to non-Lagrangian equations of motion endowed with a compatible presymplectic structure. The presence of a presymplectic structure is known to be much less restrictive for dynamics than the existence of a Lagrangian.

\section{Variational tricomplex of a gauge system}
This section provides a brief glossary on local gauge systems in the formalism of variational tricomplex. For a  more coherent exposition of these concepts we refer the reader to \cite{BARNICH2000439, Kaparulin:2011xy, Anderson1992, Dickeybook}.  Throughout the paper, we systematically use the notation and terminology of  \cite{Sharapov:2016sgx}. 
 
\paragraph{Classical fields.} In modern language {\it classical fields} are sections of a locally trivial fiber bundle $\pi: E\rightarrow M$ over a spacetime manifold $M$. The typical fiber of $E$ is called the {\it target space of fields}. The space of all field configurations is thus identified with the   space of smooth sections $\Gamma(E)$.  In this paper, we restrict ourselves to the case where $E$ is a $\mathbb{Z}$-graded vector bundle over $M$. 
As is customary in the physical literature, we will refer to this $\mathbb{Z}$-grading as the {\it ghost number} and denote the degree of a homogeneous object $A$ by $\mathrm{gh}(A)$. The Grassmann parity of fields (which governs the signs) is given by the ghost number modulo two.  In physical terms this means that we restrict ourselves
to gauge theories without fermionic degrees of freedom. The extension of our results to general theories with bosonic and fermionic fields is straightforward.

\paragraph{Variational bicomplex.} According to the principle of spacetime {\it locality}, the classical dynamics of fields are governed by partial differential equations.  The jet-bundle formalism offers then a  natural geometric framework for formulating and studying  local field theories. A relevant jet-bundle for our considerations is the bundle of infinite jets $\pi_\infty: J^\infty E\rightarrow M$ associated with the vector bundle $\pi: E\rightarrow M$. 
Each section $\varphi$ of $E$ induces a section $j^\infty \varphi$ of $J^\infty E$ by the following rule. If $E|_U=U\times \mathbb{R}^N$ is a trivializing chart with local coordinates $(x^i,\phi^a)$, then $(x^i, \phi^a,\phi^a_i, \phi^a_{ij},\ldots)$ are local coordinates in $J^\infty E|_U$ and the induced section $j^\infty \varphi: M\rightarrow J^\infty E$ is defined  by 
\begin{equation}
    x\mapsto (x, \varphi^a(x), \partial_i\varphi^a(x),\partial_i\partial_j\varphi^a(x),\ldots )\,.
\end{equation}
The section $j^\infty \varphi$ is called the $\infty$-jet prolongation of $\varphi$. 

The total space of $J^\infty E$, defined through the inverse limit $\lim\limits_{\leftarrow}J^kE$ of finite-dimensional jet-bundles, inherits the structure of a $\mathbb{Z}$-graded manifold. 
Let $\Lambda (J^\infty E)=\lim\limits_{\rightarrow }\Lambda(J^kE)$ denote the algebra of differential forms on $J^\infty E$.  It is known that the de Rham complex of $\Lambda(J^\infty E)$ splits naturally into a bicomplex for the {\it vertical differential} $\delta$ and {\it horizontal differential} $d$, so that 
\begin{equation}
    \delta^2=0\,,\qquad d^2=0\,,\qquad \delta d+ d\delta=0\,,
\end{equation}
with $\delta+d$ being the exterior differential in $\Lambda (J^\infty E)$. In the adapted coordinates above, every form on $J^\infty E$ can be written as a finite 
sum of homogeneous forms 
\begin{equation}
    f  \delta\phi^{a_1}_{I_1}\wedge\cdots\wedge \delta\phi^{a_p}_{I_p}\wedge dx^{i_1}\wedge\cdots\wedge dx^{i_q} \in \Lambda^{p,q}(J^\infty E)\,, 
\end{equation}
where $f$ is a smooth function on $J^\infty E|_U$ and $I=i_1i_2\cdots i_{|I|}$ denotes the multi-index of order $|I|$. The numbers $p$ and $q$ are called  the {\it vertical} and {\it horizontal} degree of a form, respectively. Since the vertical differential $\delta: \Lambda^{p,q}(J^{\infty}E)\rightarrow \Lambda^{p+1,q}(J^\infty E)$ implements 
the action of the variational derivative on fields, one refers to the bicomplex $\Lambda^{\ast,\ast}(J^\infty E; \delta, d)$ as {\it variational}. 

\paragraph{Classical BRST differential.} An evolutionary vector field\footnote{Recall that a vertical vector field $X$ is called evolutionary if $i_Xd + (-1)^{\mathrm{gh}(X)}di_X = 0$,
where $i_X$ is the operation of contraction of $X$ with differential forms. The Lie algebra of evolutionary vector fields will be denoted by $\mathfrak{X}_{\mathrm{ev}}(J^{\infty}E)$.} $Q$ on $J^\infty E$ is called {\it homological} if 
\begin{equation}\label{QQ}
    [Q,Q]=0\,,\qquad \mathrm{gh}(Q)=1\,.
\end{equation}
In what follows we will use the symbol $\delta_Q$ to denote the Lie derivative along $Q$. Clearly, $\delta_Q^2=0$. Moreover, the operator $\delta_Q$ anticommutes with the differentials $d$ and $\delta$: 
\begin{equation}
    \delta_Q d+d\delta_Q=0\,,\qquad \delta_Q \delta +\delta \delta_Q=0\,.
\end{equation}
The equalities follow from Cartan's formula for the Lie derivative and the fact that the vector field $Q$ is evolutionary. 
This yields the {\it variational tricomplex} $\Lambda^{\ast,\ast,\ast}(J^\infty E; \delta, d, \delta_Q)$, where the third differential
\begin{equation}
    \delta_Q: \Lambda^{p,q,r}(J^\infty E)\rightarrow \Lambda^{p,q,r+1}(J^\infty E)
\end{equation}
increases the ghost number by one. Under certain `properness'  conditions (see e.g. \cite{Kaparulin:2011xy}) the  operator $\delta_Q$ 
is called the {\it classical BRST differential}. The corresponding homological vector field $Q$ carries all the information about a  classical gauge system. In particular, the true  field configurations are determined by the stationary points of $Q$. Under the standard regularity assumptions the latter form a graded submanifold $\Sigma^\infty\subset J^{\infty}E$ defined through the inverse limit.  By definition, a field $\varphi\in \Gamma(E)$ satisfies the classical equations of motion iff $j^\infty \varphi\in \Sigma^\infty$. 

In physics, the submanifold $\Sigma^\infty\subset J^\infty E$ is usually referred  to as the {\it shell}.
By definition, the shell is invariant under the action of the homological vector field $Q$.  The restriction of the variational bicomplex $\Lambda^{\ast,\ast}(J^\infty E; \delta, d)$ on $\Sigma^\infty$ gives the {\it on-shell bicomplex} $\Lambda^{\ast,\ast}(\Sigma^\infty; \delta, d)$. The latter inherits the additional grading by the ghost number. In general, the on-shell bicomplex is not $d$-exact even locally and this gives rise to interesting invariants associated with gauge dynamics. The most notable among them are the cohomology groups $H^{0,\ast}(\Sigma^\infty; d)$ in ghost number zero. These are  known as the {\it characteristic cohomology} of a gauge system \cite{TSUJISHITA19913, 10.2307/2152923, PhysRevLett.77.4109,BARNICH2000439,  BARNICH20023,   Verbovetsky:1997fi}.  The study of natural algebraic structures on characteristic cohomology is the main subject of the present paper.  

\paragraph{Presymplectic structure.} Another geometric entity present in most gauge theories is called a {\it presymplectic structure}. This is given by a form $\omega \in \Lambda^{2,m}(J^\infty E)$ obeying the condition 
\begin{equation}
    \delta \omega \simeq 0\,.
\end{equation}
Hereinafter the sign $\simeq$ means equality modulo $d$-exact forms. Notice that the horizontal degree of a presymplectic form may take any value in the interval $0\leq m\leq \dim M$. We also impose no restriction on the ghost number of $\omega$. 
Two presymplectic forms are considered equivalent if they differ by a $d$-exact form.   By abuse of notation, we will not distinguish between a presymplectic form $\omega$ and its equivalence class in $\Lambda^{2,m}(J^\infty E)/d\Lambda^{2,m-1} (J^\infty E)$.  As we are dealing  with a vector bundle $E$,  the relative `$\delta$ modulo $d$' cohomology appears to be trivial in positive vertical degree \cite[Sec.19.3.9]{Dickeybook}. Among other things this means that each presymplectic structure has a $\delta$-exact representative, i.e., there exists $\theta\in\Lambda^{1,m}(J^\infty E) $ such that $\omega\simeq \delta \theta$. The form $\theta$ is called a {\it presymplectic potential}. In what follows we will always work with $\delta$-exact representatives $\omega=\delta\theta$ of presymplectic forms, so that $\delta\omega=0$.

An evolutionary vector field $X$ is called {\it Hamiltonian} relative to  $\omega$ if $L_X \omega\simeq 0$.  Again, the triviality of relative $\delta$-cohomology implies that 
\begin{equation}\label{XH}
    i_X\omega\simeq \delta \alpha
\end{equation}
for some $\alpha\in \Lambda^{0,m}(J^\infty E)$. It is natural to refer to the form $\alpha$ as Hamiltonian or as a Hamiltonian of the vector field $X$. 
Notice that Eq. (\ref{XH}) defines $\alpha$ up to a $d$-exact form; hence, two Hamiltonians $\alpha$ and $\alpha'$ are considered equivalent if $\alpha\simeq \alpha'$.  The space of all Hamiltonian $m$-forms is a graded Lie algebra with the  bracket
\begin{equation}\label{LB}
   \{\alpha,\beta\}=(-1)^{\mathrm{gh}(X_\alpha)}i_{X_\alpha}i_{X_\beta}\omega
\end{equation}
of degree $-\mathrm{gh}(\omega)$. Here $X_\alpha$ and $X_\beta$ are Hamiltonian vector fields with the Hamiltonians $\alpha$ and $\beta$.  The bracket enjoys the symmetry property
\begin{equation}
    \{\alpha,\beta\}\simeq -(-1)^{(\mathrm{gh}(\alpha)-\mathrm{gh}(\omega))(\mathrm{gh}(\beta)-\mathrm{gh}(\omega))}\{\beta,\alpha\}
\end{equation}
and obeys the Jacobi identity
\begin{equation}
    \{\gamma,\{\alpha,\beta\}\}   \simeq \{\{\gamma,\alpha\},\beta\}+(-1)^{(\mathrm{gh}(\gamma)-\mathrm{gh}(\omega))(\mathrm{gh}(\alpha)
    -\mathrm{gh}(\omega))}\{\alpha,\{\gamma,\beta\}\}\,.
\end{equation}
For more details see \cite[Prop. 2.1]{Sharapov_2015}. 

\paragraph{Gauge systems and their descendants.} 
By a {\it gauge system} on $J^\infty E$ we mean a pair $(Q,\omega)$ composed of a homological vector field $Q$ and a $Q$-invariant presymplectic form $\omega$ of type $(2,m)$. In other words, the vector field $Q$ is supposed to be Hamiltonian relative to $\omega$, so that $\delta_Q\omega \simeq 0$.   The last relation is equivalent to  
\begin{equation}\label{Qw}
    \delta_Q\omega=d\omega_1
\end{equation}
for some $\omega_1\in \Lambda^{2,m-1} (J^\infty E)$. We will refer to the horizontal degree of the form $\omega$ as the {\it degree of a gauge system} $(Q,\omega)$. Applying the differentials $\delta $ and $\delta_Q$ to both sides of (\ref{Qw}) and using the acyclicity of  $d$ in positive vertical degree, we find that
$$
\delta \omega_1\simeq 0 \,,\qquad \delta_Q\omega_1\simeq 0\,.
$$
Hence, $\omega_1$ is a $Q$-invariant presymplectic form of type $(2,m-1)$. Again, without loss of generality we may assume $\omega_1=\delta\theta_1$ for some presymplectic potential $\theta_1\in \Lambda^{1,m-1}(J^\infty E)$. We call the pair $(Q,\omega_1)$ the {\it descendent gauge system}. Iterating the above construction once and again, one can produce a sequence of gauge systems
$(Q,\omega_k)$ where the $k$-th presymplectic structure $\omega_k\in \Lambda^{2,m-k}(J^\infty E)$ is the descendant  of $\omega_{k-1}$. The minimal $k$ for which $\omega_k\simeq 0$ is called the {\it length of a gauge system}.  

\paragraph{Symmetries.} An evolutionary vector field $X$ is called a {\it symmetry} of a gauge system $(Q,\omega)$ if it commutes with $Q$, i.e., $
     [X,Q]=0$.
Therefore the flow generated by  $X$ preserves the shell $\Sigma^\infty\subset J^\infty E$ mapping solutions to solutions.  

A symmetry $X$ is called {\it trivial} or {\it gauge symmetry} if their exists an evolutionary vector field $Y$ such that $X=[Q,Y]$. Clearly, the gauge symmetries form an ideal in the Lie algebra of all symmetries, so that one may regard the corresponding quotient algebra, denoted by $\mathrm{Sym}(Q)$, as the Lie algebra of nontrivial symmetries. Alternatively,
we can identify the nontrivial symmetries  $\mathrm{Sym}(Q)$ with the cohomology of the differential graded Lie algebra $(\mathfrak{X}_{\mathrm{ev}}(J^\infty E), \delta_Q)$. 

A symmetry $X$ is called Hamiltonian if $X$ is a Hamiltonian vector field, that is, $L_X\omega\simeq 0$. Finally, we say that a symmetry $X$ is {\it on-shell Hamiltonian} if 
\begin{equation}\label{HH}
    i_X\omega|_{\Sigma^{^\infty}}\simeq \delta \alpha|_{\Sigma^{^\infty}}
\end{equation}
for some $\alpha$.  Read from right to left, this relation defines an  {\it on-shell Hamiltonian form} $\alpha$. 
Writing $\mathrm{Sym}(Q,\omega)$ and $\mathrm{Sym}_\Sigma(Q,\omega)$ for the Lie algebras of Hamiltonian and on-shell Hamiltonian symmetries, respectively, we get the following sequence of subalgebras in the Lie algebra of vector fields: 
$$
\mathrm{Sym}(Q,\omega)\subset \mathrm{Sym}_\Sigma(Q,\omega)\subset \mathrm{Sym}(Q)\subset \mathfrak{X}_{\mathrm{ev}}(J^\infty E)\subset \mathfrak{X}(J^\infty E)\,.
$$

When dealing with on-shell Hamiltonian symmetries and forms it is convenient to introduce the following equivalence relation
on the space of differential forms $\Lambda(J^\infty E)$:
\begin{equation}\label{ERR}
    \alpha \approx \alpha' \quad \Leftrightarrow\quad  \alpha|_{\Sigma^{^\infty}}\simeq \alpha'|_{\Sigma^{^\infty}}. 
\end{equation}
Then Eq. (\ref{HH}) takes the form $i_X\omega\approx \delta \alpha$. Although the last equation implies $L_X\omega\approx 0$, the converse is not always true as the on-shell bicomplex may not be $\delta$-exact even for vector bundles. 
Notice that Eq. (\ref{HH}) defines $\alpha$ only modulo $d$-exact and on-shell vanishing forms. Therefore, it makes sense to consider the equivalence classes of on-shell Hamiltonian forms defined by Rel. (\ref{ERR}). Due to the regularity condition, this implies the existence of forms $\beta$ and $\gamma$ such that $\alpha-\alpha'=i_Q\beta+d\gamma$. It is significant that the equivalence classes of on-shell Hamiltonian forms constitute a Lie algebra, which we denote by $\Lambda^H(Q,\omega)$, for the same Lie bracket (\ref{LB}). The last fact is quite obvious as we can identify the on-shell Hamiltonian forms with the Hamiltonian forms on $\Sigma^{\infty}$ endowed with the induced presymplectic structure $\omega|_{\Sigma^{^\infty}}$.

\paragraph{Conservation laws.} A form $\alpha\in \Lambda^{0,m}(J^\infty E)$ is said to define a {\it conservation law} of degree $m$ if 
\begin{equation}\label{da}
    d \alpha|_{\Sigma^{^\infty}}=0\,.
\end{equation}
A conservation law is called trivial if $\alpha\approx 0$. In other words, the space of nontrivial conservation laws of degree $m$ is identified with the cohomology group $H^{0,m}(\Sigma^{^\infty}; d)$ of the on-shell bicomplex. Due to the standard regularity conditions on $Q$, Eq. (\ref{da})  implies the existence of a form $\chi\in \Lambda^{1,m+1}(J^\infty E)$ such that 
\begin{equation}
    d\alpha=i_Q\chi\,.
\end{equation}
The form $\chi$ is called the {\it characteristic} of a conservation law $\alpha$.  

Let $C\subset M$ be an $m$-cycle, $\alpha$ conservation law of degree $m$, and $\varphi\in \Gamma(E)$ a solution to the field equations, then the functional
\begin{equation}
    I[\varphi]=\int_C(j^\infty \varphi)^\ast (\alpha)
\end{equation}
depends only on the homotopy class of $C$ in $M$ and is called the {\it conserved charge}.  

\paragraph{Flat gauge systems.} We say that a homological vector $Q$ field on $J^\infty E$ is {\it flat } if $i_Q \alpha|_M=0$ for all $\alpha \in \Lambda(J^\infty E)$.  Geometrically,
this means that the submanifold  $M\subset J^\infty E$, identified with the zero section, belongs to the zero locus of $Q$. Similarly, a form $\alpha \in \Lambda(J^\infty E)$ is called flat if $\alpha|_M=0$. The flat forms constitute an ideal in the exterior algebra $\Lambda(J^\infty E)$ and a subcomplex in the variational tricomplex whenever $Q$ is flat.  Let us denote the latter by $\Lambda_{\mathrm{^{flat}}}^{\ast,\ast,\ast}(J^\infty E; \delta,d,\delta_Q)$. The horizontal, vertical, and relative cohomology of the corresponding flat bicomplex $\Lambda_{\mathrm{^{flat}}}^{\ast,\ast}(J^\infty E; \delta,d)$ are given by the groups
\begin{equation}
\begin{array}{rl}
H_{^{\mathrm{flat}}}^{\ast,p}(J^{\infty}E; d)=0\,,\quad p<n\,,&\qquad H_{^{\mathrm{flat}}}^{\ast,n}(J^{\infty}E; d)\simeq \Lambda^{\ast,n}_{\mathrm{^{flat}}}(J^\infty E)/d\Lambda^{\ast,n-1}_{\mathrm{^{flat}}}(J^\infty E)\,,\\[3mm]
H_{^{\mathrm{flat}}}^{\ast,\ast}(J^{\infty}E; \delta)=0\,,&\qquad 
H_{^{\mathrm{flat}}}^{\ast,\ast}(J^{\infty}E; \delta/d)=0\,,
\end{array}
\end{equation}
$n$ being the dimension of the base manifold $M$. The proof can be found in \cite[Ch.19]{Dickeybook}. Amongst the elements of $H_{^{\mathrm{flat}}}^{0, n}(J^{\infty}E; d)$ are the equivalence classes of Lagrangians.   We will say that $(Q,\omega)$ is a {\it flat gauge system} if $Q$ is flat.  From the physical viewpoint, flat gauge systems correspond to field theory models without external sources.

\section{Lie algebra of conservation laws}\label{S3}

The presence of the zero section $j^{\infty}(0): M\rightarrow J^\infty E$ together with the canonical projection $\pi_\infty: J^\infty E\rightarrow M$ allows us to split the complex of purely horizontal forms into the direct sum
\begin{equation}\label{flat}
\begin{array}{c}
    \Lambda^{0,\ast}(J^\infty E; d)=\Lambda^{0,\ast}_{^\mathrm{flat}}(J^\infty E;d)\bigoplus \pi^\ast_\infty\Lambda^\ast(M;d)\,.
    \end{array}
\end{equation}
For flat gauge systems, this results in the natural spitting of the cohomology groups 
\begin{equation}\label{HCL}
\begin{array}{c}
    H^{0,\ast}(\Sigma^{\infty};d)=H^{0,\ast}_{^\mathrm{flat}}(\Sigma^{\infty}; d)\bigoplus H^\ast(M;d)
    \end{array}
\end{equation}
associated with the equivalence classes of conservation laws. The second summand in (\ref{HCL}) owes its existence to the topology of  $M$ rather than field dynamics. 
Therefore, in the sequel, we will ignore it, focusing on the subspace of {\it flat conservation laws} $H^{0,\ast}_{^\mathrm{flat}}(\Sigma^{\infty}; d)$. By definition, the elements of the latter subspace vanish identically when evaluated on the zero field configuration. 

Let us denote by $\hat \alpha$ the projection of a purely horizontal form $\alpha$ on the first summand in (\ref{flat}). The elements of the second summand in (\ref{flat}), being in the kernel of the differential $\delta$, belong to the center of the Lie algebra $\Lambda^H (Q,\omega)$ of on-shell Hamiltonian forms. This gives the short exact sequence of Lie algebras 
\begin{equation}\label{ES}
    \xymatrix{0\ar[r]& \pi^\ast_\infty\Lambda(M)\ar[r]^-{i}&\Lambda^{H}(Q,\omega)\ar[r]^-{p}&\hat\Lambda{}^{H}(Q,\omega)\ar[r]&0}\,,
\end{equation}
where $i$  is the natural embedding and $p$ is the canonical projection induced by the assignment  $\alpha \mapsto \hat \alpha$. The exact sequence does not split in general, in which case $\Lambda^{H}(Q,\omega)$ is a nontrivial central extension of $\hat\Lambda{}^{H}(Q,\omega)$.

With the definitions above we  are ready to formulate a generalization of  Noether's first theorem to gauge systems. 

\begin{theorem}[\cite{Sharapov:2016sgx}]\label{th31}
Let $(Q,\omega)$ be a flat gauge system of degree $m$.  Then  each Hamiltonian symmetry $X$ gives rise to a uniquely defined sequence  $\alpha_1,\alpha_2,\ldots, \alpha_m$ of flat conservation laws. Trivial symmetries generate trivial conservation laws. 
\end{theorem}

\begin{proof} Let $\alpha_0$ denote the Hamiltonian of the symmetry $X$, then 
\begin{equation}\label{a0}
    i_X\omega\simeq \delta \alpha_0\,.
\end{equation}
Applying $\delta_Q$ to this equality yields  $\delta \delta_Q\alpha_0\simeq 0$. Since the form $\delta_Q\alpha_0$ is flat,  we can write it as 
\begin{equation}\label{aa}
\delta_Q\alpha_0=d\alpha_1
\end{equation}
for some $\alpha_1\in \Lambda^{0,m-1}(J^\infty E)$. By definition, $\alpha_1$ is a conservation law with characteristic $\delta\alpha_0$. It follows from Eq. (\ref{aa}) that $d\delta_Q\alpha_1=0$. Since $\delta_Q \alpha_1$ is flat, the triviality of the group $H^{0,m-1}_{^\mathrm{flat}}(J^\infty E; d)$ implies the existence of a form $\alpha_2$ such that $\delta_Q \alpha_1=d\alpha_2$. Hence, $\alpha_2$ is a conservation law with characteristic $\delta \alpha_1$. Iterating this construction once and again, we get the desired sequence of conservation laws $\alpha_1, \alpha_2,\ldots,\alpha_m$. These are not claimed to be all nontrivial; quite the opposite: if the $k$-th conservation law happens to be trivial,  $\alpha_k\approx 0$, then so are all its descendants $\alpha_{k+1}, \ldots, \alpha_m$.   

If now $X=[Q,Y]$ for some Hamiltonian vector field $Y$, then $$\delta\alpha_0\simeq i_X\omega\simeq \delta_Qi_Y\omega\simeq -\delta \delta_Q\beta\,,$$ where $\beta$ is the Hamiltonian of $Y$. The group $H^{0,m}_{^\mathrm{flat}}(J^\infty E;\delta/d)$ being trivial, $\alpha_0=-\delta_Q\beta+d\gamma$ for some $\gamma\in \Lambda^{0,m-1}(J^\infty E)$. Eq. (\ref{aa}) then gives $d(\delta_Q\gamma-\alpha_1)=0$. Since the group $H^{0,m-1}_{^\mathrm{flat}}(J^\infty E;d)$ is zero, we conclude that  $\alpha_1\approx 0$ and the conservation law $\alpha_1$ is trivial together with all its descendants.   
\end{proof}

It is worth noting that the assumption of flatness is rather technical; in many cases, one can replace it with other mild conditions that `kill' the de Rham cohomology of the base manifold $M$.  In particular, the above theorem holds for any nonflat $Q$ whenever $M$ is contractible.

\begin{theorem}[\cite{Sharapov:2016sgx}]\label{th32}
Let $\{\alpha_k\}$ be a sequence of conservation laws associated with a Hamiltonian symmetry $X$. Then the form $\alpha_k$ is on-shell Hamiltonian relative to the $k$-th descendent presymplectic structure $\omega_k$. 
\end{theorem}

\begin{proof}
As above, we let $\alpha_0$ denote the Hamiltonian of $X$. Then Eq. (\ref{a0}) implies that 
\begin{equation}
    i_X\omega=\delta \alpha_0+d\alpha'_0
\end{equation}
for some $\alpha'_0\in \Lambda^{1,m-1}(J^\infty E)$. Acting by $\delta_Q$ on both sides of the equality, we find
\begin{equation}\label{eq1}
    (-1)^{\mathrm{gh}(X)}i_X\delta_Q\omega= \delta \delta_Q\alpha_0+d\delta_Q\alpha'_0\,.
\end{equation}
Combining this with Eqs. (\ref{Qw}) and (\ref{aa}), we get 
$
    di_X\omega_1=d\delta\alpha_1-d\delta_Q\alpha'_0
$. The acyclicity  of $d$ in positive vertical degree implies that
\begin{equation}\label{om1}
       i_X\omega_1=\delta\alpha_1-\delta_Q\alpha'_0+d\alpha'_1
\end{equation}
for some $\alpha_1'\in\Lambda^{1,m-2}(J^\infty E)$. Therefore, $i_X\omega_1\approx \delta \alpha_1$ and the form $\alpha_1$ is on-shell Hamiltonian relative to the descendent presymplectic structure $\omega_1$. Applying the operator $\delta_Q$ once again, we obtain from (\ref{om1})
\begin{equation}
      (-1)^{\mathrm{gh}(X)}i_X\delta_Q\omega_1= \delta \delta_Q\alpha_1+d\delta_Q\alpha'_1\,.
\end{equation}
This equation coincides in form with (\ref{eq1}). Hence, there is  a form $\alpha'_2$ such that 
\begin{equation}\label{om2}
       i_X\omega_2=\delta\alpha_2-\delta_Q\alpha'_1+d\alpha'_2
\end{equation}
and the conservation law $\alpha_2$ is on-shell Hamiltonian w.r.t. $\omega_2$. Proceeding in such a way, we obtain the sequence of relations 
$i_X\omega_k=\delta\alpha_k-\delta_Q\alpha'_{k-1}+d\alpha'_k$, which imply that all forms $\alpha_k$ are on-shell Hamiltonian. 

\end{proof}

\begin{theorem}\label{th33}
Let $\{\alpha_k\}$ be the sequence of flat conservation laws associated with a Hamiltonian symmetry $X$, then the assignment 
$X\mapsto\alpha_k$ defines a Lie algebra homomorphism 
\begin{equation}\label{h}
h_k: \mathrm{Sym}(Q,\omega)\rightarrow \hat \Lambda^H(Q,\omega_k)\,.
\end{equation}
\end{theorem}
\begin{proof} For every pair of Hamiltonian symmetries $X, Y\in \mathrm{Sym}(Q,\omega)$ there corresponds 
a Hamiltonian symmetry $[X,Y]$ satisfying the equation 
$$
i_{[X,Y]}\omega\simeq \delta \{\alpha_0,\beta_0\}\,.
$$
Here $\alpha_0$, $\beta_0$, and $\{\alpha_0,\beta_0\}=(-1)^{\mathrm{gh}(X)}i_Xi_Y\omega$ are the Hamiltonians of $X$, $Y$, and $[X,Y]$, respectively. 
Let $\{\alpha_k\}$ and $\{\beta_k\}$ denote the sequences of conservation laws associated with $X$ and $Y$. By definition, $$\delta_{Q}\alpha_k=d\alpha_{k+1}\,,\qquad \delta_{Q}\beta_k=d\beta_{k+1}\,,\quad k=0,1,\ldots,m\,.$$ Applying the Lie derivative $L_X$ to the second sequence of equations, we get 
\begin{equation}
    \delta_Q L_{X}\beta_k=dL_X\beta_{k+1}\quad \Rightarrow \quad \delta_Q\{\alpha_k,\beta_k\}=d\{\alpha_{k+1},\beta_{k+1}\}\,,
\end{equation}
where $\{\alpha_k,\beta_k\}=(-1)^{\mathrm{gh}(X)}i_Xi_Y\omega_k$. For $k>0$, this yields a uniquely defined sequence $\{\alpha_k,\beta_k\}$ of flat conservation laws associated with the Hamiltonian symmetry $[X,Y]$.  
\end{proof}

The homomorphism (\ref{h}) defines a Lie subalgebra $\mathrm{Im}\,h_k\subset \hat\Lambda^H(Q,\omega_k)$. By the {\it $k$-th Lie algebra of conservation laws}, denoted by $\mathrm{CL}_k(Q,\omega)$, we will understand the preimage of $\mathrm{Im}\,h_k$ in $\Lambda^H(Q,\omega_k)$, that is, $\mathrm{CL}_k(Q,\omega)=p^{-1}(\mathrm{Im}\,h_k)$. In view of (\ref{ES}), the Lie algebra $\mathrm{CL}_k(Q,\omega)$ is nothing more than a central extension of $\mathrm{Im}\,h_k$. 
A necessary condition for the existence of a nontrivial central extension of $\mathrm{Im}\,h_k$ is that $H^{m-k}(M;d)\neq 0$. 

One more version of  Noether's first theorem is given by the next statement. 

\begin{proposition}\label{P34}
Let $(Q,\omega)$ be a regular gauge system. Suppose the following two conditions are satisfied:
\begin{itemize}
    \item[{\rm (i)}] $d\omega|_{\Sigma^{^\infty}}=0$,
    \item[\rm (ii)] $H_{\rm flat}^{0,\ast}(\Sigma^{\infty}; \delta)=0$;
\end{itemize}
then every on-shell Hamiltonian form is a conservation law and Eq. (\ref{HH}) establishes a homomorphism $h: \mathrm{Sym}_\Sigma(Q,\omega)\rightarrow \Lambda^H(Q,\omega)$  of the Lie algebras. 
\end{proposition} 

\begin{proof} The second part of the statement is obvious. 
To prove the first consider  an on-shell Hamiltonian vector field $X$ with Hamiltonian $\alpha\in \Lambda^H(Q,\omega)$. Applying the horizontal 
differential $d$ to both sides of Eq. (\ref{HH}), we get 
\begin{equation}
    \delta d\alpha|_{\Sigma^{^\infty}}=(-1)^{\mathrm{gh} X}i_X d\omega|_{\Sigma^{^\infty}}\,.
\end{equation}
Regularity implies that $d\omega =i_Q\beta+\delta_Q\gamma$ for some forms $\beta$ and $\gamma$, whence 
\begin{equation}
    i_Xd\omega=i_Qi_X\beta -(-1)^{\mathrm{gh} X}\delta_Qi_X\gamma\,.
\end{equation}
Here we used the equality $[Q,X]=0$. Therefore,  $\delta (d\alpha)|_{\Sigma^{^\infty}}=0$. Condition (ii) then means that
$d\alpha|_{\Sigma^{^\infty}}=0$, i.e.,  $\alpha$ is a flat conservation law. 
\end{proof}

Notice that all descendent presymplectic structures $\omega_k$ meet condition (i), see Eq. (\ref{Qw}). Condition (ii) is fulfilled  e.g. for  homogeneous linear systems of field equations; in that case, the components of the homological vector field $Q$ are homogeneous linear functions of the vertical coordinates $\phi^a_{ij\cdots}$.

\begin{example}

All the notions  above are best illustrated by the example of Chern--Simons theory. Let us consider a $U(1)$-vector bundle  over a three-dimensional manifold $M$. For simplicity, we assume that the bundle is  trivial. Then the affine space of all $U(1)$-connections  is naturally isomorphic to the space of $1$-forms $\Lambda^1(M)$, that is, the sections of the cotangent bundle $T^\ast M$.  
In the Batalin--Vilkovisky (BV) formalism\footnote{aka {\it field-antifield formalism}, see e.g. \cite{BARNICH2000439}.}, the field spectrum of abelian Chern--Simons theory includes the connection $1$-form  $A\in \Lambda^1(M)$, ghost field $C\in \Lambda^0(M)$ and their antifields $A^\ast\in \Lambda^2(M)$ and $C^\ast\in \Lambda^3(M)$.  By definition, 
\begin{equation}
    \mathrm{gh}(C)=1\,,\qquad  \mathrm{gh} (A)=0\,,\qquad  \mathrm{gh}(A^\ast)=-1\,,\qquad  \mathrm{gh}(C^\ast)=-2\,.
\end{equation}
The space of fields and antifields carries  the canonical symplectic structure 
\begin{equation}\label{w-ChS}
    \omega =\delta A\wedge \delta A^\ast +\delta C\wedge \delta C^\ast \,,\qquad \mathrm{gh}(\omega)=-1\,,
\end{equation}
of top horizontal degree. Here $\delta$ stands for the usual variational differential on fields, whose properties are identical to those of the vertical differential above, hence the notation. The action of the classical BRST differential is given by 
\begin{equation}\label{ChS-Q}
    \delta_Q C=0\,,\qquad \delta_Q A=dC\,,\qquad \delta_QA^\ast=dA\,,\qquad \delta_Q C^\ast=dA^\ast\,.
\end{equation}
As is seen, only flat connections belong to the stationary surface of $Q$.  The zero-curvature equation $dA=0$ enjoys the standard gauge symmetry $\delta_\varepsilon A=d\varepsilon$, which  manifests itself through the second equation in (\ref{ChS-Q}). The Hamiltonian of the homological vector field $Q$ is given by the BV master Lagrangian of Chern--Simons theory:
\begin{equation}\label{L-ChS}
    i_Q\omega\simeq \delta L\,,\qquad L=\frac12 A\wedge dA+dC\wedge A^\ast\,.
\end{equation}
Thus, Eqs. (\ref{w-ChS}, \ref{ChS-Q}) define a gauge system of degree three.  

The canonical symplectic structure (\ref{w-ChS})  generates the full sequence of descendants:   
\begin{equation}
    \begin{array}{lll}
        \delta_Q\omega=d\omega_1\,,\qquad\qquad &\omega_1=\frac12\delta A\wedge \delta A+\delta C\wedge \delta A^\ast\,,\qquad \qquad &\mathrm{gh}(\omega_1)=0\,,\\[3mm]
        \delta_Q\omega_1=d\omega_2\,,\qquad &\omega_2=\delta C\wedge \delta A\,,\quad & \mathrm{gh}(\omega_2)=1\,,\\[3mm]
        \delta_Q\omega_2=d\omega_3\,,\qquad &\omega_3=\frac12\delta C\wedge \delta C\,,\quad&\mathrm{gh}(\omega_3)=2\,.
    \end{array}
\end{equation}
Hence, the length of Chern--Simons theory is four. Notice that $\omega_1$ coincides with the canonical presymplectic structure associated with the Lagrangian (\ref{L-ChS}). The classical BRST differential enjoys a nontrivial global symmetry $Y$ in ghost number $-1$, whose action on fields is given by  
\begin{equation}
 L_Y C=1\,,\qquad L_Y A=0\,,\qquad  L_YA^\ast=0\,,\qquad L_Y C^\ast  =0\,.
\end{equation}
The variational vector field $Y$ is Hamiltonian, $i_Y\omega =\delta C^\ast$, and gives rise to the sequence of conservation laws $A^\ast$, $A$, and $C$. The latter fact is readily seen  from (\ref{ChS-Q}).  The corresponding Lie brackets read
\begin{equation}
    \{A^\ast, A^\ast\}_1=0\,,\qquad  \{A, A\}_2=0\,,\qquad  \{C, C\}_3=1\,.
\end{equation}
The last equality exemplifies the phenomenon of central extension:  the abelian subalgebra $[Y,Y]=0$ of global symmetries gets a nontrivial central extension $1\in H^0(M;d)$ when evaluated on the corresponding conservation laws. 
Among the conservation laws above only the $1$-form $A$ has ghost number zero. Integrating it over a loop $\gamma\subset M$ yields the conserved charge $I[A]=\oint_\gamma A $, which is nothing but the holonomy invariant of the flat connection $A$. 

\end{example}
\section{Leibniz algebras of symmetries and conservation laws}
Most of the gauge systems encountered in physics depend on some numerical parameters such as masses of particles, coupling constants, etc.  This motivates us to consider families of gauge systems rather than isolated systems. For the sake of simplicity we will restrict ourselves to one-parameter families $(Q_t,\omega_t)$.  Also, it is convenient to be a bit sloppy about the class of functions of $t$. The homological vector field $Q_t$, for instance, may be a smooth function of $t$ or a formal power series 
\begin{equation}\label{Qt}
    Q_t=Q_0+tQ_1+t^2Q_2+\cdots\,.
\end{equation}
Such formal expansions are at the heart of perturbation theory:   If one regards  $t$ as a coupling constant, then the leading term $Q_0$ defines a `free gauge system' $(Q_0,\omega_0)$, while the higher order terms $Q_1, Q_2, \ldots$ describe a `consistent  interaction'. In many instances, the corresponding presymplectic form $\omega_t$ does not depend on $t$ at all. 

Even though the gauge system $(Q_t,\omega_t)$ varies `smoothly' with $t$, this may not be the case with the  corresponding cohomology groups.  
For example, the inclusion of interaction may violate some nontrivial symmetries of a free gauge system. It is the differences in cohomology for different $t$'s that give rise to interesting algebraic constructions that we consider below. 

\paragraph{Symmetries.} 
Differentiating the defining equality $[Q_t,Q_t]=0$, we obtain
\begin{equation}\label{QdQ}
  [Q_t,\dot Q_t]=0\,,\qquad  [\dot Q_t, \dot Q_t]=-[Q_t,\ddot Q_t] \,,
\end{equation}
where the overdots stand for the $t$-derivatives. The first equation says that the evolutionary vector field 
$\dot Q_t$ is $Q_t$-invariant; hence, it defines an odd symmetry of the gauge system $(Q_t,\omega_t)$.  According to the second equation, the symmetry $\dot Q_t$ always squares to a trivial symmetry. 
Thus, our first observation is that any one-parameter family of gauge systems possesses a canonical (perhaps trivial) symmetry  generated by $\dot Q_t$. Furthermore, this symmetry makes  $\mathrm{Sym}(Q_t)$
into a differential graded Lie algebra with the deferential 
\begin{equation}
    \partial_t \tilde X=\widetilde{[\dot Q_t, X]}\,, \qquad \partial_t^2=0\,.
\end{equation}
Here $\tilde{X}$ denotes the equivalence class of the symmetry generated by a $Q$-invariant vector field $X\in \mathfrak{X}_{\mathrm{ev}}(J^\infty E)$. 

The main purpose of this section is to equip the space of symmetries $\mathrm{Sym}(Q_t)$ with the structure of a { graded Leibniz algebra}. The latter is defined as follows. 

\begin{definition}

A {\it graded Leibniz algebra} is a graded vector space $V=\bigoplus_{n\in \mathbb{Z}} V_n$ together with a bilinear product
\begin{equation}
    \circ: V_n\otimes V_m\rightarrow V_{n+m}
\end{equation}
satisfying the Leibniz identity 
\begin{equation}\label{LI}
    a\circ (b\circ c)=(a\circ b)\circ c+(-1)^{|a||b|}b\circ (a\circ c)
\end{equation}
for all $a,b,c\in V$; here $|a|$ denotes the degree of a homogeneous element $a\in V$.
\end{definition}

It is convenient to split the $\circ$-product into its graded symmetric and skew-symmetric parts\footnote{In the following, we will often omit the word `graded'.}:
\begin{equation}
\begin{array}{c}
    a\circ b=\lb a,b\rb+\lbb a, b\rbb \,,\\[5mm]
    \displaystyle \lb a,b\rb=\frac12\big (a\circ b-(-1)^{|a||b|}b\circ a\big)\,,\qquad \lbb a,b\rbb=\frac12\big (a\circ b+(-1)^{|a||b|}b\circ a\big )\,.
    \end{array}
\end{equation}
Then it follows from the Leibniz identity (\ref{LI}) that the subspace $I\subset V$ spanned by all elements of the form  $\lbb a,b\rbb$ is an ideal of the Leibniz algebra $(V,\circ)$. Moreover, $I\circ V=0$.  The condition $I=0$ is clearly equivalent to the skew-symmetry of the $\circ$-product. In that case, Rel. (\ref{LI}) boils down to the Jacobi identity for the Lie bracket $\lb a,b\rb =a\circ b$. Therefore,
a graded Leibniz algebra with skew-symmetric product  is the same as a graded Lie algebra and vice versa.
In general,  the  skew-symmetric part of the $\circ$-product does not satisfy the Jacobi identity, but it always induces a Lie bracket in the quotient space $V/I$.  

Central to our construction is a functor 
\begin{equation}\label{F}
    \mathbf{F}: \mathbf{dgLie}\rightarrow \mathbf{gLeib} 
\end{equation}
from the category of differential graded Lie algebras to that of graded Leibniz algebras. It is defined in terms of the so-called {\it derived bracket} \cite{Kosmann_Schwarzbach_2004}.  Given a differential graded Lie algebra $L=\bigoplus_{n\in \mathbb{Z}}L_n$ with bracket $[-,-]$ of degree $n$ and differential $\partial$ of degree $1$, we write 
\begin{equation}\label{DB}
    a\circ b =(-1)^{|a|-n-1}[\partial a, b]\,.
\end{equation}
It is straightforward to check the $\circ$-product makes the graded vector space $L[-n-1]$ into a Leibniz algebra\footnote{For every $\mathbb{Z}$-graded vector space $V=\bigoplus V_n$ and $k\in \mathbb{Z}$, $V[k]$ is a graded vector space with $V[k]_n=V_{n+k}$.}. One can also see that the symmetric part
of the $\circ$-product (\ref{DB}) is in the image of the differential:
\begin{equation}\label{sDB}
    \lbb a,b\rbb=\frac12(-1)^{|a|-n-1}\partial [a,b]\,.
\end{equation}
As a consequence, the skew-symmetric part of $\circ$ induces a Lie bracket in the quotient space $L/\partial L$.  Furthermore, the differential $\partial$ defines a homomorphism of the Leibniz algebras $(L[-n-1],\circ)$ and $(L,[-,-])$. 

Applying the construction of the derived bracket (\ref{DB}) to our situation, we can turn the differential graded Lie algebra of symmetries $(\mathrm{Sym}(Q_t), \partial_t)$ into a graded Leibniz algebra w.r.t. the product 
\begin{equation}\label{lp}
    \tilde X\circ_t \tilde Y=(-1)^{\mathrm{gh}(X)-1}\widetilde{[[\dot Q_t, X], Y]}=\widetilde{[[ X, \dot Q_t], Y]}\,,
\end{equation}
$X$ and $Y$ being $Q_t$-invariant evolutionary vector fields. We will denote this Leibniz algebra by $\mathbf{Sym}(Q_t)$; as a graded vector space $\mathbf{Sym}(Q_t)=\mathrm{Sym}(Q_t)[-n-1]$. In general, the product (\ref{lp}) is not skew-symmetric and $\mathbf{Sym}(Q_t)$ is not a Lie algebra. 
An important particular situation when it does degenerate to a Lie algebra structure is the following. Let $\mathcal{L}\subset \mathrm{Sym}(Q_t)$ be a {\it commutative} subalgebra in the Lie algebra of symmetries such that $[\partial \mathcal{L},\mathcal{L}]\subset \mathcal{L}$. Then $\mathbf{F}\mathcal{L}$ is again a Lie algebra. Moreover, the derived Lie algebra $\mathbf{F}\mathcal{L}$ may well be noncommutative. 

\begin{remark}
The functor (\ref{F}) is not the only one that can be attributed to the category of differential graded Lie algebras. 
In \cite{getzler2010higher}, Getzler constructed a functor from $\mathbf{dgLie}$ to the category of $L_\infty$-algebras.  A remarkable property of Getzler's functor is that it extends the skew-symmetric part of the derived bracket (\ref{DB}), restricted to the subspace $L_{-1}$, to the full $L_\infty$-structure. The discussion of other interesting functors from the category of Leibniz algebras to that of $L_\infty$-algebras can be found in the recent paper \cite{lavau2020linftyalgebra}.  
\end{remark}

\begin{remark} Differentiating the identity $\delta_{Q_t}\omega_t=0$ w.r.t. the parameter $t$, one can see that the symmetry $\dot Q_t$ is on-shell Hamiltonian. If, as often happens, $\omega_t$ does not depend on $t$, the symmetry $\dot Q_t$ is Hamiltonian and the Leibniz product (\ref{lp}) restricts consistently onto the subspace of Hamiltonian symmetries $\mathrm{Sym}(Q_t,\omega)$. In that case, one can speak of the Leibniz algebras of Hamiltonian symmetries and the corresponding conservation laws. 
\end{remark}

\paragraph{From Lie to Leibniz and back.}

Let us return to the interpretation of the expansion (\ref{Qt}) as a formal deformation of the free gauge system  $(Q_0,\omega_0)$ by interaction. Evaluating (\ref{QdQ}) at $t=0$, we find
\begin{equation}\label{4Q}
    [Q_0,Q_1]=0\,,\qquad [Q_1,Q_1]=-2[Q_0,Q_2]\,.
\end{equation}
These equalities say that the first-order interaction $Q_1$ defines an odd symmetry whose square is a trivial symmetry of the free system.   Hence, $\tilde Q_1\in \mathrm{Sym}(Q_0)$. Formula (\ref{lp}) makes then the space of free symmetries into the Leibniz algebra with product 
\begin{equation}\label{prod}
    \tilde{X}\circ \tilde{Y}=\widetilde{[[X, Q_1],Y]}\,.
\end{equation}
We denote this Leibniz algebra by $\mathbf{Sym}(Q_0, Q_1)$ to emphasise that the product (\ref{prod}) on the symmetries of the free system $(Q_0,\omega_0)$ depends on the first-order interaction $Q_1$. At the level of graded vector spaces we  have the isomorphism  $\mathbf{Sym}(Q_0,Q_1)=\mathrm{Sym}(Q_0)[-1]$. 
As is often the case, the inclusion of interaction breaks some symmetries of the free system, and those that survive form a subalgebra $\mathcal{L}_{\mathrm{int}}$ in the Lie algebra $\mathrm{Sym}(Q_0)$. The quotient $ \mathcal{L}_{\mathrm{br}}=\mathrm{Sym}(Q_0)/\mathcal{L}_{\mathrm{int}}$, identified with the space of broken symmetries, carries no natural Lie algebra structure unless $\mathcal{L}_{\mathrm{int}}$ is an ideal in $\mathrm{Sym}(Q_0)$. 

From the viewpoint of Leibniz algebras, the unbroken symmetries constitute a {\it central ideal} in $\mathbf{Sym}(Q_0, Q_1)$, that is, \begin{equation}
    \mathbf{Sym}(Q_0,Q_1)\circ \mathcal{L}_{\mathrm{int}}\subset \mathcal{L}_{\mathrm{int}}\,,\qquad  \mathcal{L}_{\mathrm{int}}\circ \mathbf{Sym}(Q_0, Q_1)=0\,.
\end{equation}
In view of Rel. (\ref{sDB}), this includes the ideal ${I}\subset\mathcal{L}_{\mathrm{int}}$ spanned by the unbroken symmetries of the form 
\begin{equation}
\lbb\tilde{X}, \tilde Y\rbb=\frac12(-1)^{\mathrm{gh}(X)-1}\widetilde{[Q_1,[X,Y]]}\,, 
\end{equation}
$X$ and $Y$ being arbitrary  $Q_0$-invariant evolutionary vector fields.  As a result, the Leibniz product (\ref{prod}) canonically induces a Lie bracket on the vector space $ \mathcal{L}_{\mathrm{br}}[-1]$, the suspended space of broken symmetries. 

Let us summarize the results of this subsection by the following thesis. Given a consistent first-order interaction $Q_1$, we can split (noncanonically) the symmetries of the corresponding  free system $(Q_0,\omega_0)$ into two groups -- those that preserve  $Q_1$ and those that do not -- and make them both into graded Lie algebras. We emphasize that the Lie brackets in $\mathcal{L}_{\mathrm{int}}$
and $\mathcal{L}_{\mathrm{br}}[-1]$ are essentially different: the former is just the restriction of that in $\mathrm{Sym}(Q_0)$, while the latter, being induced  by the Leibniz product (\ref{prod}), involves the interaction $Q_1$. For free gauge systems, the Lie algebra of symmetries $\mathrm{Sym}(Q_0)$ is normally commutative in nonzero ghost number \cite[Sec. 3.9]{BARNICH20023} and so is its subalgebra $\mathcal{L}_{\mathrm{int}}$.  By contrast, the Lie algebra of broken symmetries $\mathcal{L}_{\mathrm{br}}[-1]$ may well be nonabelian due to the interaction $Q_1$.   Notice that whenever  $\mathrm{Sym}(Q_0)$ is abelian the corresponding Leibniz algebra $\mathbf{Sym}(Q_0,Q_1) $ is a Lie algebra with $\mathcal{L}_{\mathrm{int}}[-1]$ belonging to its centre.

\paragraph{Conservation laws.} The same construction of the derived bracket (\ref{DB}) allows one to make the Lie algebra of on-shell Hamiltonian forms $\Lambda^H(Q_0,\omega_0)$ into a Leibniz algebra. We can proceed from the simple observation that each first-order interaction $Q_1$ defines an on-shell Hamiltonian symmetry. Indeed, evaluating the equation $i_{Q_t}\omega_t\simeq \delta H_t$ at first order in $t$, we readily find
\begin{equation}\label{414}
    i_{Q_1}\omega_0 \simeq \delta H_1-i_{Q_0}\omega_1\approx \delta H_1\,.
\end{equation}
Hence, $Q_1$ is an on-shell Hamiltonian symmetry of the free gauge system $(Q_0, \omega_0)$. It follows from the second 
relation in (\ref{4Q}) that the Lie bracket (\ref{LB}) of the on-shell Hamiltonian form $H_1$ with itself is equivalent to zero:
\begin{equation}
    \{H_1,H_1\}\simeq -2\{H_0,H_2\}=2i_{Q_0}i_{Q_2}\omega_0\approx 0\,. 
\end{equation}
Notice that the degree of the bracket is opposite to the ghost number of the presymplectic structure.  Considering now the adjoint action of $H_1$, we turn  $\Lambda^H(Q_0,\omega_0)$ into a differential graded Lie algebra with the differential $\partial$ defined by the relation
\begin{equation}
    \partial \alpha=\{ H_1, \alpha\}\qquad \forall \alpha \in \Lambda^H(Q_0,\omega_0)\,.
\end{equation}
Then the functor (\ref{F}) gives us immediately the Leibniz product 
\begin{equation}\label{LH}
    \alpha \circ \beta =(-1)^{\mathrm{gh}(\alpha)+\mathrm{gh}(\omega_0)-1}\{\{H_1,\alpha\},\beta\}=\{\{\alpha, H_1\},\beta\}
\end{equation}
for all $\alpha,\beta\in \Lambda^H(Q_0,\omega_0)$. With the definition of an on-shell Hamiltonian form (\ref{HH}) we can rewrite this product in several equivalent ways:
\begin{equation}
    \alpha \circ \beta\approx  (-1)^{\mathrm{gh}(\alpha)+\mathrm{gh}(\omega_0)}\{L_{Q_1}\alpha,\beta\}\approx i_{[Q_1,X]}i_{Y}\omega\approx L_{[Q_1,X]}\beta\,.
\end{equation}
Here $X$ and $Y$
are the on-shell Hamiltonian symmetries associated with $\alpha$ and $\beta$. We will denote this Leibniz algebra by 
$\ml^H(Q_0,Q_1,\omega_0)$. As usual the passage from the Lie algebra $\Lambda^H(Q_0,\omega_0)$ to the Leibniz algebra $\ml^H(Q_0,Q_1,\omega_0)$ implies the shift in degree of on-shell Hamiltonian forms.  

The above considerations apply then to all the descendent gauge systems $(Q_0,\omega_0^k)$ associated with the free gauge system $(Q_0,\omega_0)$.  
Indeed, replacing $\omega_0$ and $\omega_1$ in (\ref{414}) with $\omega_0^k$ and $\omega_1^k$, respectively, we readily conclude that the vector field $Q_1$ is on-shell Hamiltonian relative to $(Q_0,\omega_0^k)$, that is, 
\begin{equation}
    i_{Q_1}\omega^k_0 \simeq \delta H^k_1-i_{Q_0}\omega^k_1\approx \delta H^k_1\,,
\end{equation}
$H_1^k$ being the Hamiltonian. The derived bracket construction gives then the sequence of Leibniz products 
\begin{equation}\label{LH3}
    \alpha \circ_k \beta =\{\{\alpha,H^k_1\}_k,\beta\}_k\,, \quad k=1,2,\ldots, m,
\end{equation}
for all $\alpha, \beta\in \Lambda^H (Q_0,\omega_0^k)$. We denote the corresponding Leibniz algebras by $\ml^H(Q_0,Q_1,\omega_0^k)$. 
If the assumption is made that $H^{0,\ast}_{\rm flat}(\Sigma^{\infty};\delta)=0$, the elements of these algebras become conservation laws of the free gauge system, see Proposition \ref{P34}.

In special, but not rare, instances where the presymplectic structure of the family $(Q_t,\omega_t)$ does not depend on $t$, the vector field $Q_1$ is Hamiltonian relative to $\omega_0=\omega_t$. Being a symmetry, it generates the sequence of conservation laws $H^k_1\in \mathrm{CL}_k(Q_0,\omega_0)$, $k=1,\ldots, m$. In that case the product (\ref{LH3}) restricts onto the subspace $\mathrm{CL}_k(Q_0,\omega_0)\subset \Lambda^H(Q_0,\omega^k_0)$ making it into a Leibniz algebra, which we denote by $\mathbf{CL}_k(Q_0,Q_1,\omega_0)$.

\section{Applications}

In this section, we exemplify the general approach developed above by two fundamental physical models: Yang--Mills theory and Einstein's gravity without matter. Both the theories enjoy asymptotic  conservation laws with nontrivial Leibniz algebras. The technical simplicity of these models combined with the Batalin--Vilkovisky (BV) formalism enables us to exhibit  explicitly all the relevant presymplectic structures, conserved currents and their algebras. As the third example we consider the gravity field subject to the so-called unimodularity condition. This model, being highly nonlinear,  is intended to demonstrate that lower-degree conservation laws are not prerogative of free gauge theories alone.

\subsection{Yang--Mills fields}

Let $M$ be a $4$-dimensional spacetime manifold, $\mathcal{G}$ a compact Lie algebra, and $\mathrm{Tr}$ an invariant nondegenerate trace on the universal enveloping  algebra $\mathcal{U}(\mathcal{G})$. Denote by $\Lambda_q^p(M,\mathcal{G})$ the space of differential $p$-forms on $M$ with values in $\mathcal{G}$ that have ghost number $q\in \mathbb{Z}$.  Recall that in the BV formalism the spectrum of Yang--Mills (YM) theory without matter consists of the following  fields and antifields: 
\begin{equation}
A \in \Lambda_0^1(M,\mathcal{G})\,, \quad C\in \Lambda_1^0(M,\mathcal{G})\,, \quad A^{\ast} \in \Lambda_{-1}^3(M,\mathcal{G})\,, \quad C^{\ast} \in \Lambda_{-2}^4(M, \mathcal{G})\,. 
\end{equation}
The space of fields and antifields carries the canonical symplectic structure\footnote{Hereinafter the wedge product combines the exterior product of forms with the product in $\mathcal{U}(\mathcal{G})$.}
\begin{equation}\label{YMomega}
    \omega = \mathrm{Tr} \left( \delta A\wedge \delta A^{\ast } + \delta C\wedge \delta C^{\ast} \right) \, , \qquad \mathrm{gh}(\omega) = -1 \, .
\end{equation}
The $2$-form $\omega$ being nondegenerate, the corresponding homological vector field $Q$ is uniquely defined by the relation $i_Q\omega=\delta \mathcal{L}$, where the Hamiltonian form $\mathcal{L}$ is given by the standard BV extension of the YM Lagrangian, namely,  
\begin{equation}\label{YM}
    \mathcal{L} = \mathrm{Tr} \left(\frac12 F \wedge \tilde F + A^{\ast} \wedge DC +  g  C^{\ast } CC\right)\,.
\end{equation}
Here $\tilde F$ is the Hodge dual of the curvature $2$-form $F=dA+ gA\wedge A$ and $D$ stands for the covariant derivative, e.g. $DC=dC+g[A,C]$. Regarding the YM coupling constant $g$ as deformation parameter, we can write 
\begin{equation}\label{LYM}
    \mathcal{L}= \mathcal{L}_0+ g\mathcal{L}_1+g^2 \mathcal{L}_2\,,
\end{equation}
and similar expansion takes place for the homological vector field $Q=Q_0+gQ_1+g^2Q_2$. One can easily check that $Q^2=0$ or, what is the same, $\{\mathcal{L},\mathcal{L}\}\simeq 0$, where the braces stand for the canonical BV bracket associated with the symplectic structure (\ref{YMomega}). 
The leading term $\mathcal{L}_0$ describes the dynamics of free YM fields, while $\mathcal{L}_{1}$ and $\mathcal{L}_2$ introduce a consistent interaction.  

Applying the BRST differential $\delta_Q$ to the symplectic structure (\ref{YMomega}), one can find the following descendants: 
\begin{equation}\label{om-om}
    \begin{array}{lll}
        \delta_{Q} \omega = d\omega_1  \, , \qquad& \omega_1 = \mathrm{Tr} ( \delta A \wedge \delta \tilde F+ \delta C\wedge \delta A^{\ast} ) \,, \quad\qquad & \omega_1 \in \Lambda^3(M) \, , \qquad \mathrm{gh}(\omega_1)=0\,,\\ [3 mm]
        \delta_{Q} \omega_1 = d \omega_2 \, , \quad\qquad  & \omega_2 = \mathrm{Tr} ( \delta C\wedge \delta \tilde F) \,, & \omega_2 \in \Lambda^2(M)\,,\qquad \mathrm{gh}(\omega_2)=1 \,.
    \end{array}
\end{equation}
Since $\delta_Q\omega_2=0$, the length of YM theory is equal to $3$.   Unlike (\ref{YMomega}), the descendent presymplectic forms $\omega_1$ and $\omega_2$ depend on the coupling constant $g$. The first descendant is nothing but the canonical presymplectic structure associated with the Lagrangian (\ref{YM}), that is, $\omega_1=\delta \theta_1$, where the presymplectic potential $\theta_1$ comes from the variation 
    $\delta \mathcal{L}=i_Q\omega +d\theta_1$. The second descendant $\omega_2$ gives rise to the Lie and Leibniz algebra structures on lower-degree conservation laws (surface charges).
    
    In order to define these algebraic structures we note that $Q_1$ is a Hamiltonian vector field generating a symmetry of the free Lagrangian $\mathcal{L}_0$. By Theorem \ref{th31}, it yields a pair of conservation laws $\mathbf{J}_1$ and $\mathbf{J}_2$ defined by
\begin{equation}
    \begin{array}{llll}
     \delta_{Q_0} \mathcal{L}_1 = d \mathbf{J}_1 \, ,& \quad \mathbf{J}_1 = \mathrm{Tr} \big( \tilde F_0\wedge [A,C]+A^{\ast} C  C\big)\,,& \quad \mathbf{J}_1 \in \Lambda^3(M) \, ,&\quad \mathrm{gh}(\mathbf{J}_1)=1\,,\\[3mm]
        \delta_{Q_0} \mathbf{J}_1 = d \mathbf{J}_2 \, , & \quad \mathbf{J}_2 = \mathrm{Tr} \big( \tilde F_0 CC \big)\,, & \quad \mathbf{J}_2 \in \Lambda^2(M)\, ,&\quad \mathrm{gh}(\mathbf{J}_2)=2\,,\\[3mm]
        \delta_{Q_0}\mathbf{J}_2=0\,,&&&
    \end{array}
\end{equation}
$F_0=dA$ being the strength of free YM fields. In the calculations above we used the following formulas for the action of the free BRST differential on fields and antifields:
\begin{equation}
 \delta_{Q_0} C^\ast=dA^\ast\,,\qquad \delta_{Q_0} A^\ast=d\tilde{F}_0\,,\qquad \delta_{Q_0}A=dC\,,\qquad \delta_{Q_0}   C=0\,.
\end{equation}
With the help of $\mathbf{J}$'s we can turn the Lie algebras of on-shell Hamiltonian forms of degree three and two into a pair of Leibniz algebras by setting
\begin{equation}\label{YMLeibniz}
    \alpha\circ_k \beta = \{ \{\alpha, \mathbf{J}_k \}_k, \beta \}_k\,.
\end{equation}
Here the braces with the subscript $k=1,2$ stand for the Lie brackets determined by the presymplectic structures (\ref{om-om}). 

The $2$-form $\mathbf{J}_2$ is not the only surface current that one can attribute to  free YM fields. 
The free Lagrangian $\mathcal{L}_0$ obviously enjoys the shift symmetry $C \rightarrow C + \epsilon \xi$, where $\xi$ is a vector of $\mathcal{G}$ and $\epsilon$ is a constant parameter of ghost number one ($d\epsilon=0$).  By  Noether's theorem, we get immediately the conserved current 
\begin{equation}\label{j0}
        J_1^\xi = \mathrm{Tr}(\xi A^{\ast}) \in \Lambda^3(M)\,, \qquad \mathrm{gh}(J_1^\xi) = -1 \,,
\end{equation}
and its descendant 
\begin{equation}\label{j1}
J_2^\xi =    \mathrm{Tr}(\xi \tilde {F}_0)\in \Lambda^2(M)\,, \qquad \mathrm{gh}(J_2^\xi) = 0 \,.
\end{equation}
Obviously,
\begin{equation}
    dJ^\xi_1\approx 0\,,\qquad dJ^\xi_2\approx 0\,, \qquad \delta_{Q_0}J_1^\xi=dJ_2^\xi\,,\qquad \delta_{Q_0}J^\xi_2=0\,. 
\end{equation} The conserved current $J^\xi_2$, being of ghost 
number zero,  admits a straightforward  physical interpretation. If $S$ is a closed space-like surface in $M$, then the integral 
\begin{equation}
    q^\xi=\int_S J_2^\xi
\end{equation}
defines the net (color) charge enclosed by the surface $S$. The surface currents (\ref{j0}) and (\ref{j1}) form the abelian Lie algebras 
\begin{equation}
    \{J^\xi_k,J_k^{\xi'}\}_k=0\,,\qquad k=1, 2,
\end{equation}
as is usually the case for free theories. The first-order interaction $\mathcal{L}_1$ breaks the shift symmetry above, leading thus to  nontrivial Leibniz products of surface currents:
\begin{equation}
    J_k^\xi\circ_k J_k^{\xi'}=J_k^{[\xi,\xi']}\,,\qquad \forall \xi,\xi'\in \mathcal{G}\,,\qquad k=1,2.
\end{equation}
As is seen, both the products are skew-symmetric and define the Lie algebra structure isomorphic to $\mathcal{G}$. In such a way we are able to 
reproduce the color Lie algebra at the level of surface currents.  We emphasize that the currents $J_k^\xi$ are conserved on the free equations of motion 
$Q_0=0$ and cannot be promoted to conservation laws of full YM theory.  Nevertheless, to make them into  a  nonabelian Lie algebra 
we need the cubic interaction vertices accommodated
in $\mathcal{L}_1$.  These vertices are essentially responsible for the `nonabelian part' of the gauge generators. 

In order to present a genuine example of Leibniz algebra which is not Lie, we note that the 
free Lagrangian $\mathcal{L}_0$ is invariant under orthogonal transformations. Indeed, as the Lie algebra $\mathcal{G}$ is supposed to be compact, there is a basis $\{t_a\}_{a=1}^n$ in $\mathcal{G}$ such that $\mathrm{Tr}(t_at_b)=\delta_{ab}$. Expanding now the fields in terms of this basis, e.g. 
$A=A^at_a$, we see that the quadratic in fields Lagrangian $\mathcal{L}_0$ is determined by the Euclidean metric $\delta_{ab}$. As a result,  it appears to be invariant under 
the $O(n)$-rotations of fields: $A^a\rightarrow A'^a=R^a{}_b A^b$ and similar transformations  for the other fields; here $R=(R^a{}_b)$ is an orthogonal matrix from $O(n)$. Applying Noether's theorem to infinitesimal rotations gives the following conserved currents for the free YM fields: 
\begin{equation}
        j^{ab}_1 =  A^{[a} \wedge \tilde F_0^{b]} + C^{[a} A^{\ast \,\, b]} \,, \qquad  j^{ab}_1 \in \Lambda^3(M) \,,\qquad \mathrm{gh}(j^{ab}_1)=0\,.
\end{equation}
These have descendants of the first generation:
\begin{equation}
   \delta_{Q_0}j^{ab}_1=dj_2^{ab}\,,\qquad  j^{ab}_2= C^{[a} \tilde F_0^{b]} \,, \qquad j^{ab}_2 \in \Lambda^2(M)\,,\qquad \mathrm{gh}(j_2^{ab})=1\,. 
\end{equation}
As usual the square brackets stand for skew-symmetrization of indices.  By Theorem \ref{th33}, the Lie brackets of the currents reproduce the commutation relations of the Lie algebra $o(n)$: 
\begin{equation}
    \{j_k^{ab}, j_k^{cd} \}_k = \delta^{bc} j_k^{ad} - \delta^{ac} j_k^{bd} + \delta^{ad} j_k^{bc}- \delta^{bd} j_k^{ac}   \,,\qquad k=1, 2.
\end{equation}
The currents $j^{ab}_k$ generate the Hamiltonian action of $o(n)$ on the space of fields and antifields. 
Evaluating now the symmetric part of the Leibniz products $j_k^{ab}
\circ_k j_k^{cd}$, one easily finds
\begin{equation}\label{jjj}
     \lbb j^{ab}_k, j^{cd}_k \rbb_k = -\frac12\{\mathbf{J}_k,\{j_k^{ab},j_k^{cd}\}_k\}_k=\delta^{bc} {J}^{ad}_k- \delta^{ac} {J}_k^{bd}   + \delta^{ad} {J}^{bc}_k - \delta^{bd} {J}^{ac}_k \,,
\end{equation}
where $J^{ab}_k=-\frac12\{\mathbf{J}_k, j^{ab}_k\}_k$ are new conserved currents of ghost number $k$.
Generally the structure constants of the Lie algebra $\mathcal{G}$, entering the currents $\mathbf{J}_k$, are  not $O(n)$-invariant, so that the new currents $J^{ab}_k$ are different from zero. 

Writing the currents  (\ref{j0}) and (\ref{j1}) in terms of the basis above, $J^\xi_k=\delta_{ab}J^a_k\xi^b$, we obtain
\begin{equation}\label{jjjj}
\begin{array}{rcl}
    j^{ab}_k\circ_k J^c_k &=& \displaystyle f^{ac}{}_{d}\, j^{db}_k - f^{bc}{}_{d}\, j^{da}_k +   \delta^{ac}f^b_{de}\, j_k^{de} - \delta^{bc}f^a_{de}\, j_k^{de}\,,\\[3mm]
    J^c_k\circ_k j^{ab}_k &=& f^{bc}{}_{d}\, j^{da}_k - f^{ac}{}_{d}\, j^{db}_k\,,\qquad k=1,2.
    \end{array}
\end{equation}
These Leibniz products also have nonzero symmetric parts.

\subsection{Einstein gravity}
In the vierbein formalism, the dynamics of the gravity field are described by ten $1$-form fields: the vierbein  $e^a$ and the 
spin-connection $w^{ab}=-w^{ba}$. As usual, we use the Minkowski metric $\eta_{ab}$ for raising and lowering the Lorentz indices $a,b=0,1,2,3$.  
In order to have explicit control over the gauge symmetries corresponding to the diffeomorphisms of the spacetime manifold $M$ and local Lorentz invariance, one extends the field content with the 
ghost fields $c^a$ and $c^{ab}=-c^{ba}$ as well as the corresponding  antifields labeled by star. Table \ref{T1} collects the full spectrum of fields and antifields together with their form degrees and ghost numbers.  

\begin{table}[h!]
\centering
    \begin{tabular}[t]{c|c|c|c|c|c|c|c|c}
       &  $c^{*}_a$ & $c^{*}_{ab}$ & $e^{*}_a$ & $w^{*}_{ab}$ & $e^a$ & $w^{ab}$ & $c^a$ & $c^{ab}$ \\
    \hline
    $\mathrm{deg}$ & 4 & 4 & 3 & 3 & 1 & 1 & 0 & 0 \\ 
     \hline
     $\mathrm{gh}$ & -2 & -2 & -1 & -1 & 0 & 0 & 1 & 1 \\
\end{tabular}
\caption{Fields and antifields in the vierbein formulation of gravity.}
\label{T1}
\end{table}
The fields and antifields are canonically conjugate to each other w.r.t. the BV symplectic structure
\begin{equation}\label{Gom}
    \omega = \delta e^{*}_a \wedge \delta e^a + \delta w^{*}_{ab} \wedge \delta w^{ab} + \delta c^{*}_a \wedge \delta c^a + \delta c^{*}_{ab} \wedge \delta c^{ab}\,.
\end{equation}
The BV master Lagrangian of Einstein's gravity  now reads  
\begin{equation}\label{GL}
    \mathcal{L} = \frac12 \epsilon_{abcd} e^a \wedge e^b \wedge R^{cd}+ e^{\ast}_a \wedge Dc^a +  w^{\ast}_{ab} \wedge Dc^{ab} + e^{\ast}_a \wedge c^a{}_b e^b + c^{\ast}_a c^a{}_b c^b + \frac12 c^{\ast}_{ab} c^a{}_c c^{cb}.
\end{equation}
Here $D=d+w$ is the Lorentz covarian differential  with the curvature $2$-form $R^{ab} = d w^{ab} + w^a{}_c \wedge w^{cb}$. 
The homological vector field $Q$ underlying Einstein's gravity is just the Hamiltonian vector field generated by the master Lagrangian (\ref{GL}) and the symplectic structure (\ref{Gom}). As is well known, see e.g. \cite{brahamsixth}, the characteristic cohomology of  pure gravity is empty. The nontrivial conservation laws arise only upon linearization of the Lagrangian (\ref{GL}) about a suitable geometric background. Consider, for simplicity, the flat background geometry with a vierbein $h^a$ obeying $dh^a=0$. Then we can put $e^a = h^a + \tilde e^a$, where the $1$-forms  $\tilde e^a$ describe fluctuations over the flat background.  On substituting this decomposition into (\ref{GL}), we get\footnote{Rescaling all the fields and antifields by $g$, while multiplying  $\mathcal{L}$ by $g^{-2}$, we  can bring the Lagrangian (\ref{L012}) into the form (\ref{LYM}) with $g$ playing the role of coupling constant. } 
\begin{equation}\label{L012}
    \mathcal{L}\simeq \mathcal{L}_0+\mathcal{L}_1+\mathcal{L}_2\,,
\end{equation}
where 
\begin{equation}\label{Lagranglingrav}
\begin{array}{rcl}
       \mathcal{L}_0 &=& \epsilon_{abcd} dw^{ab} \wedge \tilde e^c \wedge h^d + \frac12 \epsilon_{abcd} w^a{}_e \wedge w^{eb} \wedge h^c \wedge h^d  \nonumber  + e^{*}_a \wedge (dc^a + c^{ab}h_b) + w^{*}_{ab} \wedge dc^{ab}\,,\\[3mm]
       \mathcal{L}_1&=& \frac12 \epsilon_{abcd} dw^{ab} \wedge \tilde e^c \wedge \tilde e^d + \epsilon_{abcd} w^a{}_e \wedge w^{eb} \wedge \tilde e^c \wedge h^d + e^{*}_a \wedge (w^a{}_e c^e + c^{ab}\tilde e_b) + w^{*}_{ab} \wedge w^a{}_e c^{eb} \, + \\[3mm]
        & & + \, c^{\ast}_a c^a{}_b c^b + \frac12 c^{\ast}_{ab} c^a{}_c c^{cb} \,,\\[3 mm]
       \mathcal{L}_2&=& \frac12 \epsilon_{abcd} w^a{}_e \wedge w^{eb} \wedge \tilde e^c \wedge \tilde e^d\,.
\end{array}
\end{equation}
It is convenient to introduce the following background $1$- and $2$-forms:
\begin{equation}
    h_{abc}=\epsilon_{abcd}h^d\,,\qquad H_{ab}=\epsilon_{abcd}h^c\wedge h^d\,.
\end{equation}
Then the action of the free BRST differential is given by
\begin{equation}\label{gaugetrans}
    \begin{array}{ll}
        
         \delta_{Q_0} e^a = dc^a + c^a{}_b h^b\,, \qquad\qquad& \delta_{Q_0} e^{*}_a = d w^{bc} \wedge h_{abc}\,, \\ [3 mm]
          \delta_{Q_0} w^{ab} = dc^{ab}\,, \quad & \delta_{Q_0} w^{*}_{ab} = de^c\wedge h_{abc} + H_{ac}\wedge w^c{}_{b} -H_{bc}\wedge w^c{}_a\,,\\[3mm]
         \delta_{Q_0} {c}^a = 0\,, &  \delta_{Q_0} {c}^{*}_a = de^{*}_a\,,\\[3mm]
          \delta_{Q_0} {c}^{ab}= 0\,, & \delta_{Q_0} {c}^{*}_{ab} = dw^{*}_{ab} - e^{\ast}{}_{a} \wedge h_{b}+ e^{\ast}{}_{b} \wedge h_{a}\,. 
    \end{array}
\end{equation}
The BV symplectic structure (\ref{Gom}) has the following descendants: 
\begin{equation}
\small
    \begin{array}{lll}
        \delta_{Q_0} \omega = d \omega_1\,, \qquad\qquad & \omega_1 = \delta e^a \wedge \delta w^{bc} \wedge h_{abc} + \delta c^a \wedge \delta e^{*}_a +  \delta c^{ab} \wedge \delta w^{*}_{ab} \,,\\ [3 mm]
        \delta_{Q_0} \omega_1 = d \omega_2\,, & \omega_2 = (\delta c^a \wedge \delta w^{bc} + \delta c^{ab} \wedge \delta e^c) \wedge h_{abc} \,,\\ [3 mm]
        \delta_{Q_0} \omega_2 = d\omega_3\,, & \omega_3 = \delta c^{ab} \wedge \delta c^c \wedge h_{abc} \,.
    \end{array}
\end{equation}
The free Lagrangian $\mathcal{L}_0$ is invariant under the shifts 
\begin{equation}\label{symCC}
    c^a \rightarrow c^a+\xi^a\,,\qquad c^{ab} \rightarrow c^{ab}+\xi^{ab}\,,
\end{equation}
where the transformation parameters $\xi=(\xi^a, \xi^{ab})$ obey the conditions
\begin{equation}
    \xi^{ab}=-\xi^{ba}\,,\qquad d\xi^{ab}=0\,,\qquad d\xi^a = \xi^{ab}h_b\,.
\end{equation}
In local coordinates where $h^a=dx^a$, we can solve these equations as $\xi^{a}=\zeta^a+\zeta^{ab}x_a$ and $\xi^{ab}=\zeta^{ab}$, with $\zeta^a$ and $\zeta^{ab}=-\zeta^{ba}$ being  arbitrary constant parameters. The symmetry (\ref{symCC}) is Hamiltonian relative to the BV symplectic structure (\ref{Gom}) and is generated by the Hamiltonian form 
$H=\xi^ac^\ast_a+\xi^{ab}c^\ast_{ab}$. By Noether's first theorem, we obtain the $10$-parameter family of conserved currents. We find
\begin{equation}
   \delta_{Q_0}H=dJ^{\xi}_1\,,\qquad  J^\xi_1=\xi^a e^\ast_a+\xi^{ab}w^\ast_{ab}=\zeta_a P^a_1+\zeta_{ab}M^{ab}_1\,.
\end{equation}
Applying the BRST differential (\ref{gaugetrans}) yields the following descendants:
\begin{equation}
\begin{array}{ll}
      \delta_{Q_0} J^\xi_1 = dJ_2^{\xi}\,,\qquad\qquad &J^{\xi}_2=\zeta_a P_2^a+\zeta_{ab}M^{ab}_2=(\xi^{ab} e^c+\xi^a w^{bc})\wedge h_{abc} \,, \\[3mm]
     \delta_{Q_0} J^\xi_2 = dJ_3^{\xi}\,,\quad &J^{\xi}_3=\zeta_a P_3^a+\zeta_{ab}M^{ab}_3=(\xi^{ab}c^c  + \xi^ac^{bc})\wedge h_{abc}\,, 
\end{array}
\end{equation}
and $\delta_{Q_0}J_3^\xi=0$. The conserved currents $J^\xi_2$, being $2$-forms of ghost number zero, are used to define the total energy-momentum $\mathcal{P}$ and the angular momentum $\mathcal{M}$ of an asymptotically flat universe.  These are given by the charges 
\begin{equation}
    \mathcal{P}^a=\int_S P_2^a\,,\qquad \mathcal{M}^{ab}=\int_S M^{ab}_2\,,
\end{equation}
where the integrals are over a closed space-like surface $S\subset M$ at infinity. In particular, $\mathcal{P}^0$ gives the ADM energy  \cite{PhysRevLett.77.4109, BARNICH20023}. From the physical viewpoint, it is quite  natural to identify the above charges with the generators of the Poincar\'e group -- the isometry group of flat spacetime -- and  to expect them to form the Poincare algebra w.r.t. the Lie brackets. However, this is not the case. A straightforward calculation yields 
\begin{equation}
    \{J^\xi_1,J^{\xi'}_1\}_1=0\,,\qquad \{J_2^\xi, J_2^{\xi'}\}_2=0\,,\qquad \{J_3^\xi,J_3^{\xi'}\}_3=-(\xi^a\xi'^{bc}+\xi'^a\xi^{bc})h_{abc}\,.
\end{equation}
As is seen, the abelian Lie algebra of the symmetry transformations (\ref{symCC}) gets a central extension when evaluated at the level of the $1$-form currents $J_3^\xi$. This is the phenomenon of central extension that we discussed in Sec. \ref{S3}.

In order to reproduce the commutation relations of the Poincar\'e algebra the cubic interaction of gravitons needs to be taken into account.  As usual, the cubic part of the Lagrangian (\ref{L012}) leads to a sequence of conserved currents of the corresponding free theory. Explicitly, 
\begin{equation}
\small
    \begin{array}{ll}
        \delta_{Q_0} \mathcal{L}_1 = d\mathbf{J}_1\,, \qquad\qquad & \mathbf{J}_1 =  \big(\frac12 c^a  w^b{}_d \wedge w^{dc} - w^a{}_d \wedge e^d c^{bc}\big)\wedge h_{abc} + e^{*}_a c^{ab}  c_b + \frac12 w^{*}_{ab} c^a{}_c c^{cb}\, , \\ [3mm] 
        \delta_{Q_0} \mathbf{J}_1 = d\mathbf{J}_2\,, & \mathbf{J}_2 = \big(c^a  w^b{}_d c^{dc} + \frac12 e^a c^b{}_d c^{c}\big)\wedge h_{abc} \, ,\\ [3mm]
        \delta_{Q_0} \mathbf{J}_2 = d\mathbf{J}_3\,, & \mathbf{J}_3 = \frac12 c^a c^b{}_d c^{dc}h_{abc}\,.
    \end{array}
\end{equation}
By formulas (\ref{LH}, \ref{LH3}), these currents define Leibniz algebra structures on the spaces of on-shell Hamiltonian forms of degree $3$, $2$, and $1$. Since the currents $J^\xi_k$ commute to each other modulo central elements,  the symmetric part
of the corresponding Leibniz products 
\begin{equation}
    J^\xi_k \circ_k J^{\xi'}_k = \{\{ J^\xi_k, \mathbf{J}_k\}_k, J^{\xi'}_k \}_k\,,\qquad k=1,2,3,
\end{equation} 
vanishes identically, while the skew-symmetric one reproduces the structure relations of the Poincar\'e algebra:
\begin{equation}
    \begin{array}{l}
        \lb M_k^{ab}\,, M_k^{a'b'} \rb_{k} = \eta^{aa'} M_k^{bb'} - \eta^{ba'} M_k^{ab'} + \eta^{bb'} M_k^{aa'} - \eta^{ab'} M_k^{ba'} 
        \,, \\ [3 mm]
        \lb P_k^a\,, M_k^{bc} \rb_k = \eta^{ac} P_k^{b} - \eta^{ab} P_k^{c} \,, \\ [3mm] 
        \lb P_k^a, P_k^b \rb_k = 0\,.
    \end{array}
\end{equation}
It should be recognized that the construction of this Lie algebra involves essentially the cubic part $\mathcal{L}_1$ of the Lagrangian, and not just the conserved currents of linearized gravity.

\subsection{Unimodular gravity}

Unimodular gravity provides an example of a {\it nonlinear} gauge theory with lower-degree conservation laws, see e.g. \cite{PhysRevD.40.1048,HENNEAUX1989195,Percacci_2018}.  
This time it is convenient to work in the metric instead of the vierbein   formalism we used above.  So, let  $M$ be a four-dimensional spacetime manifold endowed with a pseudo-Riemannian metric $g_{\mu\nu}$ and let $\sqrt{-g}d^4x$ denote the canonical volume form on $(M,g)$. The main idea of unimodular gravity, which goes back to Einstein,   is to impose the algebraic constraint 
\begin{equation}\label{fixg}
     \sqrt{-g} = 1 
\end{equation}
on the metric tensor. Physically, one may regard this as the partial fixing of a reference frame. 
The unimodularity condition  (\ref{fixg}) breaks the group of spacetime diffeomorphisms -- the gauge group of general relativity -- to the subgroup of volume preserving diffeomorphisms. If  $\xi$ is a vector field generating an infinitesimal gauge transformation 
\begin{equation}\label{gg}
    \delta_\xi g_{\mu\nu}=L_\xi g_{\mu\nu}=\nabla_\mu\xi_\nu+\nabla_\nu\xi_\mu\,,
\end{equation}
then the unimodularity condition requires that    
\begin{equation}\label{xi1}
    \delta_\xi\sqrt{-g}=L_\xi\sqrt{-g}=\sqrt{-g}\nabla_{\mu}\xi^\mu=0 \qquad \Rightarrow \quad \nabla_\mu \xi^\mu = 0 \,.
\end{equation}
Hence, the gauge parameter $\xi$ appears to be constrained by a differential equation. One can solve this equation for $\xi^\mu$ in terms of an arbitrary bivector  $\xi^{\mu\nu}=-\xi^{\nu\mu}$ as $\xi^\mu=\nabla_\nu\xi^{\nu\mu}$. Then the volume preserving gauge transformations take the form 
\begin{equation}\label{xi2}
    \delta_\xi g_{\mu\nu}=\nabla_\mu \nabla^\lambda \xi_{\lambda\nu}+\nabla_\nu \nabla^\lambda \xi_{\lambda\mu}\,.
\end{equation}
Unlike (\ref{gg}), these transformations appear to be reducible. Indeed, using the symmetry properties of the Riemann and Ricci tensors, one can see that the shift 
\begin{equation}\label{xx}
\xi^{\mu\nu}\rightarrow \xi^{\mu\nu}+\nabla_\lambda\xi^{\lambda\mu\nu}
\end{equation}
does not affect the r.h.s. of Eq. (\ref{xi2}) for an arbitrary $3$-vector $\xi^{\lambda\mu\nu}$; hence,  one may regard (\ref{xx}) as gauge symmetry for gauge symmetry. 
The latter transformation, in its turn, is invariant under similar redefinitions 
\begin{equation}\xi^{\lambda\mu\nu}\rightarrow \xi^{\lambda\mu\nu}+\nabla_\sigma\xi^{\sigma\lambda\mu\nu}\,,\end{equation} 
with $\xi^{\sigma\lambda\mu\nu}$ being an arbitrary $4$-vector. At this step the sequence of reducibility relations stops by the reason of dimension. 
Now, to put unimodular gravity into the standard BV formalism, one just promotes the gauge parameters $\xi^{\mu\nu\cdots}$ to the corresponding ghost fields $C^{\mu\nu\cdots}$ and introduces the conjugate antifields\footnote{For an alternative BV formulation of unimodular gravity we refer to  \cite{Kaparulin:2019gsx}.}. The full field content is represented in Table \ref{T2} below.

\begin{table}[h!]
\centering
    \begin{tabular}[t]{c|c|c|c|c|c|c|c|c}
       &  $C^{*}_{\sigma\lambda\mu\nu}$ & $C^\ast_{\lambda\mu\nu}$ & $C^\ast_{\mu\nu}$ & $g^\ast_{\mu\nu}$ & $g^{\mu\nu}$ & $C^{\mu\nu}$ & $C^{\lambda\mu\nu}$ & $C^{\sigma\lambda\mu\nu}$ \\
     \hline
     $\mathrm{gh}$ & -4 & -3 & -2 & -1 & 0 & 1 & 2 & 3 \\
\end{tabular}
\caption{Fields and antifields of unimodular gravity.}\label{T2}
\end{table}
Notice that all the ghost and antighost fields are totally skew-symmetric tensors and the antifield $g^\ast_{\mu\nu}=g^\ast_{\nu\mu}$ is supposed to be traceless, i.e., 
\begin{equation}\label{fixtraceg}
    g^{\mu\nu}g^\ast_{\mu\nu}=0\,.
\end{equation}  
A relevant presymplectic structure on  the fields and antifields is obtained by restricting the canonical symplectic form 
\begin{equation}\label{omegaUnimod}
    \omega=\big(\delta g^\ast_{ \mu\nu}\wedge \delta g^{\mu\nu} + \delta C^\ast_{\mu\nu}\wedge \delta C^{\mu\nu}+ \delta C^\ast_{\lambda\mu\nu}\wedge \delta C^{\lambda\mu\nu}+ \delta C^\ast_{\sigma\lambda\mu\nu}\wedge \delta C^{\sigma\lambda\mu\nu}\big)\sqrt{-g}d^4x
\end{equation}
onto the subspace of fields obeying the constraints (\ref{fixg}) and (\ref{fixtraceg}). Since the bracket of constraints defined by $\omega$ is clearly nonzero, 
\begin{equation}
    \big\{d^4x(\sqrt{-g}-1), d^4x\sqrt{-g}g^{\mu\nu}g^\ast_{\mu\nu}\big\}=2\sqrt{-g}d^4x>0\,,
\end{equation}
the restricted form is likewise symplectic. 

Now, the (minimal) master action for unimodular gravity can be written as  
\begin{equation}\label{UngravS}
    S = \int d^4x \sqrt{-g} \Big[R + g^{\ast}_{\mu\nu} \nabla^\mu \nabla_\lambda C^{\lambda\nu} + C^{\ast}_{\mu\nu}\nabla_{\lambda} C^{\lambda\mu\nu} + C^{\ast}_{\mu\nu\lambda} \nabla_\sigma C^{\sigma\mu\nu\lambda}  \Big] \,.
\end{equation}
Although the first term looks like the conventional Einstein--Hilbert action,
one should keep in mind the algebraic constraint (\ref{fixg}) enforced on the metric.  Evaluating the variation of the action (\ref{UngravS}) under the variation of metric subject to the condition $g^{\mu\nu}\delta g_{\mu\nu}=0$, one can see that that the corresponding equations of motion are fully equivalent to the Einstein equations with cosmological constant:  
\begin{equation}\label{RR}
    R_{\mu\nu}-\frac12g_{\mu\nu}R=-\frac14\Lambda g_{\mu\nu}\,.
\end{equation}
The only difference is that the constant $\Lambda$ is now a constant of integration rather than a parameter in the action functional. It follows from Eq. (\ref{RR}) that $R=\Lambda$, i.e., the scalar curvature $R$ represents a zero-degree conservation law.  One may wonder about a global symmetry this conservation  law comes from.  The answer is almost obvious:  the action is invariant under the transformations 
\begin{equation}\label{symmC-3}
    C^{\mu\nu\lambda\sigma} \rightarrow C^{\mu\nu\lambda\sigma} + \frac{\kappa}{\sqrt{-g}} \epsilon^{\mu\nu\lambda\sigma} \,,
\end{equation}
where $\epsilon^{\mu\nu\lambda\sigma}$ is the Levi--Civita symbol and $\kappa$ is a constant parameter of ghost number $3$. By the No\"ether theorem this gives immediately the conserved current \begin{equation}
    J_0 = C^{\ast}_{\mu\nu\lambda} dx^{\mu} \wedge dx^{\nu} \wedge dx^{\lambda} \,,\qquad \mathrm{gh}(J_0)=-3\,.
\end{equation}
Although this current on its own has no physical interpretation, it  gives rise to a sequence of lower-degree conservation laws. Indeed, applying the BRST differential generated by the master Lagrangian (\ref{UngravS}) and the symplectic structure (\ref{omegaUnimod}),  we find 
\begin{equation}
    \begin{array}{lll}
        \delta_Q J_0 = d J_{1} \,, \qquad\qquad & J_{1} = C^{\ast}_{\mu\nu} dx^{\mu} \wedge dx^{\nu} \,, \qquad\qquad &\mathrm{gh}(J_1)=-2\,,\\ [3 mm]
        \delta_Q J_{1} = d J_{2} \,, \qquad & J_{2} = \nabla^\mu g^{\ast}_{\mu\nu} dx^{\nu} \,,\qquad &\mathrm{gh}(J_2)=-1\,,\\ [3 mm]
        \delta_Q J_{2} = d J_3 \,, \qquad & J_3 = - \frac14 R\,,\qquad & \mathrm{gh}(J_3)=0\,.
    \end{array}
\end{equation}
The last nontrivial current of ghost number zero is proportional to the scalar curvature we have discussed above.

As one more source of examples of nonlinear gauge theories with lower-degree conservation laws, we would like to mention Higher Spin Gravity.  As shown in the recent paper \cite{Sharapov_2020}, the higher-spin extension of four-dimensional gravity enjoys an infinite number of nontrivial conservation laws of degrees zero and two.  
\section*{Acknowledgments}
\label{sec:Aknowledgements}
The second author is grateful to Glenn Barnich for useful discussions. This study was supported by the Tomsk State University Development Programme (Priority-2030) and by the Foundation for the Advancement of Theoretical Physics and Mathematics “BASIS”.

\bibliographystyle{plain}

\providecommand{\href}[2]{#2}\begingroup\raggedright\endgroup
\end{document}